\documentclass[12pt]{article}
\usepackage{amsmath,amssymb, amsfonts, amsthm, latexsym, hyperref, titlesec}
\usepackage{xcolor}
\usepackage[numbers]{natbib}
\newtheorem{lem}{Lemma}[section]
\newtheorem{lemma}[lem]{Lemma}

\newtheorem{thm}[lem]{Theorem}

\newcommand{\N}{\mathbb{N}}

\newcommand{\E}{\mathbb{E}}
\renewcommand{\P}{\mathbb{P}}
\renewcommand{\Pr}{\mathbb{P}}

\title{Optimal broadcasting in networks with faulty nodes}
\author{Yoel Grinshpon \and Ori Gurel-Gurevich}

\begin{document}
\maketitle

\begin{abstract}
Large computer networks are an essential part of modern technology, and quite often information needs to be broadcast to all the computers in the network. If all computers work perfectly all the time, this is simple. Suppose, however, that some of the computers fail occasionally. What is the fastest way to ensure that with high probability all working computers get the information?

In this paper, we analyze three algorithms to do so. All algorithms terminate in logarithmic time, assuming computers fail with probability $1-p$ independently of each other. We prove that the third algorithm, which runs in time $(1+o(1))(\frac{\log N}{\log(1+p)})$, is asymptotically optimal.
\end{abstract}

\section{Introduction and results}

Suppose we have a network with $N$ nodes (which stand for computers), each of them \emph{active} independently with probability $p$. Suppose also that one of these nodes has a message that needs to be conveyed to all active nodes. Each node can send one message  per time unit to any other node. Once an active node receives the message, it too can send the message to other nodes. The question we are asking is what method of spreading the message will minimize the time we have to wait until all active nodes have received the message, with high probability (that is, the probability goes to 1 as $N$ tends to infinity).

First, we analyze the \emph{naive algorithm}. In this algorithm, each informed node sends a message at every time unit to a node chosen randomly with uniform distribution.
\begin{thm} \label{naive}
For the naive algorithm, with high probability, all active nodes will receive the message after $\left(1+o(1)\right)\left(\frac{1}{\log(1+p)}+\frac{1}{p}\right)\log N$ time units.
\end{thm}

Second, we analyze the \emph{cyclic algorithm}, which was suggested by Amnon Barak \cite{AB}. This algorithm begins like the naive algorithm, with messages sent at random, for $\frac{(1+o(1))\log N}{\log(1+p)}$ time units. At this point, we can be sure that most active nodes have received the message, and each node starts sending messages to the nodes next to him in cyclic order. More precisely, if we denote the nodes by $\{1,...,N\}$ then node $i$ sends a message to node $i+1,i+2,\ldots$ etc. modulo $N$.
\begin{thm} \label{cyclic}
For the cyclic algorithm, with high probability, all active nodes will receive the message after $(1+o(1))(\frac{1}{\log(1+p)}+\frac{1}{-\log (1-p)}) \log N$ time units.
\end{thm}

Third, we introduce the \emph{improved cyclic algorithm}. This algorithm, like the cyclic algorithm, also begins with messages sent at random for $\frac{(1+o(1))\log N}{\log(1+p)}$ time units. Then, we divide the network into disjoint segments, each consisting of $\ell=\sqrt {\log N}$ consecutive nodes (in the cyclic order). A segment is \emph{good} if it has at least $\frac{\ell p}{2}$ active nodes, at least one of which is informed, and \emph{bad} otherwise. Now, each node informs all nodes in its segment. If a segment is bad, then the nodes in it stop transmitting. If a segment is good, then all active nodes in the segment (which are now informed) begin to inform the nodes in the next segments (in the cyclic order) and then the segment after it and so on. This takes $\lceil 2/p \rceil$ time units per segment, and, since bad segments are rare, we only need to do this for a short time, so this part of the algorithm takes only $o(\log N)$ time units.

\begin{thm} \label{improved}
For the improved cyclic algorithm, with high probability, all active nodes will receive the message after $(1+o(1))(\frac{\log N}{\log(1+p)})$ time units.
\end{thm}

Finally, we prove that the improved cyclic algorithm is asymptotically optimal.

\begin{thm} \label{optimal}
For any algorithm, with high probability, the number of time units required for all the active nodes to be informed is at least $(1+o(1))(\frac{\log N}{\log(1+p)})$.
\end{thm}

Note that in our model, nodes are active independently with probability $p$. Our results can be applied to the case where exactly $K$ random nodes are active by comparing to the independent model with $p$ chosen to be slightly less then $K/N$, so that the number of active nodes will be less then $K$ with high probability (see Lemma \ref{imlemma}). If the $K$ active nodes are chosen by an adversary, the naive algorithm is not affected, but for the cyclic and improved cyclic algorithms we need to apply a random permutation to the labels of the vertices and this information needs to be transmitted along with the message.

\section{Related Works}

The topic of broadcasting information to all nodes of a network has been extensively studied for many different models. For a survey of the different models, see Pelc \cite{Pelc1}. Gasieniec and Pelc \cite{GP} gave an algorithm working in $O(\log^2 n)$ (under slightly different assumptions then ours), Diks and Pelc \cite{DP} improved this to $O(\log n)$ (see also \cite{DDMM}). Results for some variation on the model can be found in \cite{A,B,C,D,G}. In all of these the analysis is up to a constant, whereas we determine the optimal running time up to $1+o(1)$. For examples of real world systems using such broadcasting algorithms, see \cite{BDLLS,F}.

Frieze and Grimmett \cite{FG} studied the running time of the naive algorithm when there are no faults ($p=1$). Their result was further refined by Pittel \cite{Pittel}. Our Theorem \ref{naive} generalizes these results to all $0\le p \le 1$. Doerr, Huber and Levavi \cite{DHL} analyzed the naive algorithm in the case of faulty links, instead of nodes, getting the same running time as in our Thoerem \ref{naive}.

\section{Preliminaries}
Suppose we have a network with $N$ nodes, with every node connected directly to the rest of the nodes, but only some of them are \emph{active}. We assume that each node is active with some fixed probability $0<p<1$, independently from the rest of the nodes. Let $n$ be the number of active nodes.

\begin{lemma} \label{imlemma}
With high probability,
\begin{equation} \label{n}
pN - N^{2/3} < n < pN + N^{2/3}
\end{equation}
\end{lemma}
\begin{proof}
$n\sim Bin(N,p)$, so $\E[n]=Np$ and $Var(n)=N\cdot p \cdot (1-p)$. By Chebyshev's inequality,
$$\Pr(|n-pN|>N^{2/3})=\Pr(|n-\E[n]|>N^{2/3})\leq \frac{N\cdot p \cdot (1-p)}{N^{4/3}}\underset{N\to \infty}{\longrightarrow}0 \ .$$
\end{proof}

Our goal is to disseminate a piece of information to all active nodes. When $t=0$, only node 0 is \emph{informed} (we assume that it is active). At every time step $t\in \N$, each informed node may choose one other node and send a message to it. If the other node is active, then it becomes informed and from time $t+1$ onwards it may also send messages and inform other nodes. The nodes do not know apriori which nodes are active, although this kind of information can be sent from node to node along with the piece of information, at no additional cost.

Let $k_t$ be the number of informed nodes at time $t$, and $T_x$ to be the first time when $k_t\geq x$. We are interested in the asymptotic behaviour of $T_n$, under different algorithms. All the running times we find are logarithmic in $N$, so our results are the form $T_n=(1+o(1))C(p)\log(N)$ with high probability. That is,
$$\P\big((C(p)-\epsilon)\log(N)\le T_n \le (C(p)+\epsilon)\log(N)\big)\to 1$$
as $N\to\infty$, for any fixed $\epsilon$ and $p$.

\section{The Naive Algorithm}
In the naive algorithm each node sends the message to an independent, uniformly random node.

Fix some $0 < \epsilon < \frac{1}{2}$. In order to analyze the running time of the algorithm, we shall divide it into three stages:
\begin{itemize}
\item Stage 1: From $t=0$ where only one node is informed, until $\epsilon pN$ nodes are informed, i.e. from time 0 to $T_{\epsilon pN}$.
\item Stage 2: From time $T_{\epsilon pN}$ to $T_{(1-\epsilon) pN}$.
\item Stage 3: From time $T_{(1-\epsilon )pN}$ to $T_n$.
\end{itemize}

We now analyze the time it takes for the naive algorithm to conclude each stage.

\begin{lemma} \label{stage1}
For any $\delta >0$, there exists $m_1\in \N $ that does not depend on $N$, such that
$$\Pr (k_{t_1+m_1}<\epsilon pN)<\delta \ .$$
\end{lemma}

\begin{proof}
Define $Z_t$ to be a Galton-Watson branching process in which every node has two children with probability $p'=p(1-\epsilon )$, and one child with probability $1-p'$. The expected number of children is $1+p'$. As long as $k_t\leq \epsilon pN$, $Z_t$ is dominated by $k_t$ since we can couple them together, and for each node in $k_t$, the probability of informing an active uninformed node at this point is larger than the probability of each node in $Z_t$ having two children. Therefore, it is enough to show that there exists such a $m\in \N $ for $Z_t$, i.e. $m$ such that $\Pr(Z_{t_1+m}<\epsilon p N)<\delta$.

Define $W_q=\frac{Z_q}{(1+p')^q}$. It is known from branching process theory (see \cite{Y} and \cite{Z} p.14) that $W_q$ is a martingale, and $W_q$ converges almost surely to some random variable $W$. Since the offspring distribution has finite support and the probability for no offsprings is zero, by Kesten-Stigum theorem we get that $\E[W]=1$ and that $\Pr(W=0)=0$. $W_q$ is a martingale, so by Doob's optional stopping theorem we get that for any $q$,
$$\Pr(W>2W_q \mid W_q)\leq \frac{1}{2} \ .$$

For $\ell, q \in \N$, define $A$ to be the event $W_q<\frac{1}{(1+p')^\ell}$, so $\Pr(W\leq \frac{2}{(1+p')^\ell} \mid A)\geq \frac{1}{2}$.
Therefore,
$$\frac{1}{2}\leq \Pr\left(W\leq \frac{2}{(1+p')^\ell} \Big| A\right)=\frac{\Pr\left(W\leq \frac{2}{(1+p')^\ell}\cap A\right)}{\Pr(A)}\leq \frac{\Pr\left(W\leq \frac{2}{(1+p')^\ell}\right)}{\Pr(A)} ,$$
so
$$\Pr(A)\leq 2\Pr\left(W\leq \frac{2}{(1+p')^\ell}\right)\ .$$

Since $\Pr(W=0)=0$, there exists $\ell \in \N$ such that $ 2\Pr(W\leq \frac{2}{(1+p')^\ell})<\delta$. Taking $q=t_1+\ell$, we get
\[
\begin{split}
\Pr\left(Z_{t_1+\ell}<(1+p')^{t_1}\right)&=\Pr\left(\frac{Z_q}{(1+p')^q}<\frac{1}{(1+p')^\ell}\right)=\Pr\left(W_q<\frac{1}{(1+p')^\ell}\right) \\
&=\Pr(A)\leq 2\Pr\left(W\leq \frac{2}{(1+p')^\ell}\right)<\delta,
\end{split}
\]
and since $(1+p')^{t_1}=\epsilon p N$ we get $\Pr(k_{t_1+\ell}<\epsilon pN)<\delta$.
\end{proof}

We proved an upper bound of roughly $\frac{\log N}{\log(1+p)}$ on the time it takes for stage 1 to conclude. The following lemma establishes a corresponding lower bound for any message sending algorithm, which also yields Theorem \ref{optimal}.

\begin{lemma} \label{lowerbound}
For any $0<a\le 1$ and for any algorithm, the probability that after $\frac{\log N}{\log(1+p)}-K$ steps there are $a p N$ informed nodes is at most $\frac{1}{a p(1+p)^K}$.
\end{lemma}
\begin{proof}
Let $0<a<1$. Let us define a new model in which each node knows which nodes already received a message and which didn't, but not which nodes are active. Furthermore, assume that all informed nodes can coordinate their message sending. In this model, clearly the optimal algorithm would be for each node to send a message to some new node and to make sure no two nodes send messages to the same node. Hence, the number of informed nodes in the optimal algorithm under this model, $Z_t$, is a branching process that has two children with probability $p$, and one child with probability $1-p$.

Clearly, the new model dominates the old so we can couple $Z_t$ and $k_t$ such that $Z_t\geq k_t$ for every $t$. Taking $t=\frac{\log N}{\log(1+p)}-K$ and using Markov's inequality, we get
$$\Pr(k_{t_0}\ge apN)\le \Pr(Z_{t_0}\geq apN)\leq \frac{\E[Z_{t_0}]}{apN}=\frac{(1+p)^{\frac{\log N}{\log(1+p)}-K}}{apN}=\frac{1}{ap(1+p)^K} \ .$$
\end{proof}

\begin{proof}[Proof of Theorem \ref{optimal}]
By Lemma \ref{lowerbound}, if we take $t=\frac{\log N}{\log(1+p)}-K(N)$ where $K(N)$ is any function which is $o(\log(N))$ and $\omega(1)$ we get that with high probability, the algorithm does not conclude before time $t$.
\end{proof}

Let us now continue with analyzing the stages of the naive algorithm.

\begin{lem} \label{stage2}
For every $\delta>0$ and $0<c_1<c_2<p$ there exists $m\in \N$ such that
$$\Pr(T_{c_2 N}-T_{c_1 N}>m)<\delta \ .$$
\end{lem}
\begin{proof}
Define $Z_t$ to be a branching process that with probability $p'=p-c_2$ has two children, and $1-p'$ to have one, so that the expected number of children is $1+p'$. In order for $Z_t$ to be compatible with $k_t$, we start $Z_t$ at time $T_{c_1 N}$ and with value $Z_{T_{c_1 N}}=cN$. As long as $k_t<c_2 N$, by coupling, $Z_t$ is dominated by $k_t$, so it is enough to show that there exists $m\in N$ such that $\Pr(Z_{T_{c_1 N}+m}< c_2 N)<\delta$.

Defining $W_q=\frac{Z_q}{(1+p')^q}$, the conditional expectation $\E[Z_{T_{c_1 N}+\ell} \mid Z_{T_{c_1 N}}]=c_1 N\cdot(1+p')^\ell$ will exceed $c_2 N$ when $(1+p')^\ell=\frac{c_2}{c_1}\Longleftrightarrow \ell=\frac{\log (\frac{c_2}{c_1})}{\log (1+p')}$. By the same argument as in Lemma \ref{stage1}, there exists $m'\in \N$ such that $\Pr(Z_{T_{c_1 N}+\ell+m'}<c_2 N)<\delta$, so taking $m=\ell+m'$ we get the wanted result.
\end{proof}

In particular, we see that for any $\delta>0$, we can choose $m_2$ such that $\Pr(T_{(1-\epsilon) p N}-T_{\epsilon p N}>m_2)<\delta$.

Define $t_3=\frac{\log N}{p(1-\epsilon)}$.

\begin{lem} \label{stage3upper}
For $\delta>0$, there exists $m_3$ such that
$$\Pr(T_n-T_{(1-\epsilon) pN}>t_3+m_3)<\delta \ .$$
\end{lem}
\begin{proof}
Define $X_m$ to be the number of uninformed nodes after $m$ steps in stage three. By lemma \ref{imlemma} we have $X_0 < 2 \epsilon p N$ with high probability.

In the third stage, $k_t>(1-\epsilon) pN$. For a node not informed yet, the probability of not being informed by a specific node is $\left(1-\frac{1}{N}\right)$, so the probability of not being informed by any of the informed nodes at a certain step is smaller than
\begin{equation}
\left(1-\frac{1}{N}\right)^{(1-\epsilon)pN} \le  e^{-p(1-\epsilon )}
\end{equation}
and after $m$ steps in stage 3, the probability of any specific node not being informed is smaller than $e^{-p(1-\epsilon)m}$, so $\E[X_m]\le 2 \epsilon p N e^{-p(1-\epsilon)m}$. Therefore, using Markov's inequality, the probability that there is at least one uninformed node is bounded by
$$\P(X_m\geq 1)\leq \E[X_m],$$
and $\E[X_m]<\delta$ will hold when
$$2 \epsilon pN e^{-p(1-\epsilon )m}<\delta\Longleftrightarrow  m>\frac{\log N}{p(1-\epsilon)} + \frac{\log(\frac{2p \epsilon}{\delta})}{p(1-\epsilon)} \ .$$

Taking $m_3= \frac{\log(\frac{2p \epsilon}{\delta})}{p(1-\epsilon)}$, we get the desired result.
\end{proof}

Denote $t'_3=\frac{\log N}{p}$.

\begin{lem} \label{stage3lower}
For $\delta>0$, there exists $m$ that does not depend on $N$, such that
$$\Pr(T_n-T_{(1-\epsilon)p N}< t'_3-m)<\delta \ .$$
\end{lem}

\begin{proof}
We shall describe a different model which dominates our model and the lemma holds for it. In the new model, all active nodes send messages, even the uninformed nodes. If an uninformed node receives a message from an uninformed node, it becomes informed. Start with $Y_0=\frac{\epsilon}{2} p N$ uninformed nodes and let $Y_k$ to be the number of uninformed nodes after $k$ steps of this model. Obviously, we can couple this new model with our original model so that the number of uninformed nodes in the original model after $k$ steps of the third stage is at least $Y_k$.

The expectation of $Y_k$ is
$$\E[Y_k]=\frac{\epsilon}{2} p N \left(1-\frac1N\right)^{pNk} = \frac{\epsilon}{2} p N  e^{-(1+o(1))pk} \ .$$

Let us calculate $Var(Y_k)$. Denote by $B$ the set of uninformed nodes at the onset (so $|B|= \frac{\epsilon pN}{2}$), and $Y_{i,k}$ to be $1$ if the $i$-th node is uninformed after $k$ steps of the new model, so that $Y_k=\underset{i\in B}{\sum}Y_{i,k}$.
$$Var(Y_k)=\E[Y_k^2]-\E[Y_k]^2=\sum_{i\in B}^{}\E[Y_{i,k}^2] + \sum_{i\neq j}^{}\E[Y_{i,k} Y_{j,k}]-\E[Y_k]^2.$$
Observe that $\Pr(Y_{j,k}=1|Y_{i,k}=1)\le \Pr(Y_{j,k}=1)$, and therefore
\[
\begin{split}
\E[Y_{i,k} Y_{j,k}]&=\Pr((Y_{i,k}=1)\cap (Y_{j,k}=1)) \\
&=\Pr(Y_{i,k}=1) \ \Pr(Y_{j,k}=1|Y_{i,k}=1)\leq \Pr(Y_{i,k}=1)^2 \ .
\end{split}
\]
Now, since $\E[Y_k]=\frac{\epsilon pN}{2}\cdot \Pr(Y_{i,k}=1)$, we get $\underset{i\neq j}{\sum}\E[Y_{i,k} Y_{j,k}]\leq (\frac{\epsilon pN}{2})^2\cdot \Pr(Y_{i,k})^2=\E[Y_k]^2$, and therefore
$$Var(Y_k)\leq \sum_{i\in B}^{}\E[Y_{i,k}^2] = \E[Y_k] \ .$$

By Chebyshev's inequality we now have
$$\Pr(Y_k=0)\leq \Pr(|Y_k-\E[Y_k]|\geq\E[Y_k])\leq\frac{Var(Y_k)}{\E[Y_k]^2}\leq \frac{1}{\E[Y_k]} \ .$$

Therefore, as long as $\E[Y_k]\geq\frac{1}{\delta}$, the probability of finishing is smaller than $\delta$. Taking $k=\frac{\log(\delta\epsilon p N/2)}{p}$ yields $\E[Y_k]\geq\frac{1}{\delta}$, so taking $0<m=-\frac{\log(\delta\epsilon p /2)}{p}$ we get the desired result.
\end{proof}

We can now prove Theorem \ref{naive}.
\begin{proof}[Proof of Theorem \ref{naive}]
Observe that $T_n=T_n-T_{(1-\epsilon)pN}+T_{(1-\epsilon)pN}-T_{\epsilon pN}+T_{\epsilon pN}$.
By Lemma \ref{stage1} we know that there exists $m_1\in \N$ such that $ T_{\epsilon pN} \leq \frac{\log N}{\log (1+p(1-\epsilon))}+m_1$ with probability larger than $1-\delta$. By Lemma \ref{stage2} we know that there exists $m_2\in \N$ such that with probability larger than $1-\delta$, $T_{(1-\epsilon)pN}-T_{\epsilon pN}>m_2$. By Lemma \ref{stage3upper} we know that there exists $m_3\in \N$ such that with probability larger than $1-\delta$, $T_n-T_{(1-\epsilon)pN}<\frac{\log N}{p}+m_3$. Taking $N\longrightarrow \infty$, and $\epsilon, \delta \longrightarrow 0$ slowly enough such that $m_1,m_2,m_3=o(\log N)$, we get that for the naive algorithm, with high probability, all active nodes will be informed by time $\left(1+o(1)\right)\left(\frac{1}{\log(1+p)}+\frac{1}{p}\right)\log N$.

By Lemma \ref{lowerbound}, we know there exists some $m_4$ such that is stage 1 is not done before $\frac{\log N}{\log (1+p)}-m_4$, with probability at least $1-\delta$. By Lemma \ref{stage3lower}, we know that there exists $m_5$ such that with probability larger than $1-\delta$, stage 3 is not done before $\frac{\log N}{p}-m_5$. Again, taking $N\longrightarrow \infty$, and $\epsilon, \delta \longrightarrow 0$ slowly enough such that $m_4$ and $m_5$ are $o(\log N)$, we get that the naive algorithm, with high probability, will not be done before time  $\left(1+o(1)\right)\left(\frac{1}{\log(1+p)}+\frac{1}{p}\right)\log N$.
\end{proof}

\section{The Cyclic Algorithm}

As Lemmas \ref{stage1}, \ref{lowerbound} and \ref{stage2} show, the first and second stages of the naive algorithm are optimal. The cyclic algorithm consists of two phases (not to be confused with the three stages of analysis).

The first phase consists of sending messages randomly, as in the naive algorithm, until some time $(1+o(1))\frac{\log N}{\log(1+p)}$ which guarantee (by Lemmas \ref{stage1} and \ref{stage2}) that $(1-\epsilon)p N$ nodes are informed, with high probability, for some fixed $\epsilon>0$. In order for all the informed nodes to know when the first phase has finished, each nodes includes the time (as measured from the start of the algorithm) in the information it sends.

In the second phase, nodes are sending messages in a cyclic order: Mark the nodes $\{1,\dots,N\}$. If node $i$ is informed at the beginning of the second phase, it will send messages to node $i+1 \mod N$, then $i+2 \mod N$, and so on.

Denote $t_5=-\frac{\log N}{\log(1-p(1-\epsilon))}$. Let $T'_n$ be the number of steps from the start of the second phase until all nodes are informed.

\begin{lem} \label{cyclicstage3}
For $\delta>0$, there exists $m$ such that $\Pr(T'_n > t_5+m)<\delta$.
\end{lem}
\begin{proof}
Observe that at the beginning of the second phase, there are, with hight probability, at least $(1-\epsilon)p N$ informed nodes. Since the first phase is invariant under all permutations, the probability that any specific $k$ nodes are all uninformed is bounded by
$$\frac{{N-k \choose (1-\epsilon)p N}}{{N \choose (1-\epsilon)p N}}\le \left(1-(1-\epsilon)p\right)^k \ .$$

Hence, if we take $k= \frac{\log \delta -\log N}{\log (1-p(1-\epsilon))}$ and using union bound, we get that the probability that there are $k$ consecutive (in the cyclic order) uninformed nodes is at most $\delta$. In other words, the probability of the second phase not concluding after $k$ steps is at most $\delta$. Taking $m=\frac{\log \delta}{\log (1-p(1-\epsilon))}$ we get the desired result.
\end{proof}

We can now prove Theorem \ref{cyclic}.
\begin{proof}[Proof of Theorem \ref{cyclic}]
As noted at the beginning of the section, by the end of the first phase, which takes $(1+o(1))\frac{\log N}{\log(1+p)}$ time steps, with high probability there are $(1-\epsilon)pN$ informed nodes, for some $\epsilon>0$. By lemma \ref{cyclicstage3}, for any $\delta>0$, there is some $m$ such that the probability of having uninformed nodes at time $-\frac{\log N}{\log(1-p(1-\epsilon))}+m$ is at most $\delta$. Taking $N\to \infty$, and $\epsilon , \delta \to 0$, slowly enough such that $m=o(\log N)$ yields the desired result.
\end{proof}

The next theorem show that this algorithm is indeed an improvement over the naive algorithm.

\begin{thm}
For every $0<p<1$, the cyclic algorithm's running time is asymptotically better than the naive algorithm's running time.
\end{thm}
\begin{proof}
We need to show that $\frac{1}{p}>\frac{1}{-log(1-p)}$ for $0<p<1$, which happens if and only if $p+\log(1-p)<0$. Define
$$f(p)=p+\log(1-p) \ .$$
Then $f(0)=0$ and the derivative
$$f'(p)=1-\frac{1}{1-p}$$
is negative for any $0<p<1$, so $f(p)<0$ for any $0<p<1$.
\end{proof}

\section{The Improved Cyclic Algorithm}

As noted earlier, the first phase of the cyclic algorithm is optimal. The improved cyclic algorithm will have the same first phase and a second phase, described below, which will take only $o(\log(N))$ steps to inform all nodes. Notice that at the end of the first phase, there are, with high probability, less then $(1-(1-\epsilon)p)$ uninformed nodes and they are invariant under all permutations of the nodes. Consider a model where, each active node is informed with probability $(1-2\epsilon)p$, independently. Call this the independent model. Comparing to the situation at the end of the first phase, in the independent model there will be more uninformed nodes, with high probability, and they are also invariant under all permutations. Since our second phase is going to be monotone, that is, changing a node from uninformed to informed will not cause the second phase to fail, it is enough to prove that it works with high probability under the independent model.

Let $\ell(N)=\sqrt{\log N}$ (in fact, any function which is $o(\log N)$ and $\omega(1)$ will do here). Group the nodes $1,2,\ldots,N$ into $N/\ell$ nonoverlapping segments of $\ell$ contiguous nodes each. Define a segment to be \emph{good} if it has at least $\ell (1-2\epsilon)p/2$ informed nodes at the end of the first phase, and \emph{bad} otherwise. Let $q(N)$ be the probability that a segment is good.

\begin{lemma} \label{q}
Under the independent model, $q(N)\to 1$ as $N\to\infty$.
\end{lemma}
\begin{proof}
Under the independent model, $X$, the number of informed nodes in a segment is Binomial with parameters $\ell$ and $p'=(1-2\epsilon)p)$. The expectation is $\E[X]=\ell p'$ and the variance is $Var(X)=\ell p' (1-p')$. By Chebyshev inequality we have
$$\P(X<\frac{\ell p'}{2}\le \P(|X-\ell p'|>\frac{\ell p'}{2}) \le \frac{\ell p' (1-p')}{(\ell p')^2} < \frac{1}{\ell p'}$$
and this bound tends to 0, since $\ell\to \infty$.
\end{proof}

At the beginning of the second phase, each informed node informs the other nodes in its segment. This takes at most $\ell$ (which is $o(\log N)$) steps, after which each node in the segment knows which other nodes are informed and which are not and specifically, whether the segment is good or bad. If the segment is bad, then the nodes stop sending messages. If the segment is good, then all informed nodes begin to inform the node of the next segment (in the cyclic order), making sure that different nodes informs different nodes in the next segment. This takes only a constant number of steps, specifically, at most $\ell/(k/2)\le \frac{3}{p}$ steps. After that, they start informing the next segment, and so on.

\begin{lemma} \label{second_phase}
With high probability $T_n-T_{(1-\epsilon)pN}$ is $o(\log N)$.
\end{lemma}
\begin{proof}
The first part of the second phase takes $\ell$ steps which is $o(\log N)$.
After that, the time it takes to finish the second phase is at most $\frac{3}{p}$ times the length of the longest sequence of bad segments. Under the independent model, the probability that there are $k$ consecutive bad segments at a specific location is $(1-q)^k$. Using union bound, the probability of getting such a sequence anywhere is at most $N(1-q)^k$. Plugging in $k=\frac{2\log N}{\log (\frac{1}{1-q})}$ we get that this probability is bounded by $\frac{1}{N}$. By Lemma \ref{q}, $q\to 1$, so our choice of $k$ is $o(\log N)$.
\end{proof}
\begin{proof}[Proof of Theorem \ref{improved}]
By Lemmas \ref{stage1} and \ref{stage2}, the first phase takes $(1+o(1))\frac{\log N}{\log(1+p)}$, with high probability. By Lemma \ref{second_phase}, the second phase takes $o(\log N)$.
\end{proof}

By Theorem \ref{optimal}, this is asymptotically optimal.

\section*{Acknowledgments}
We thank Amnon Barak for introducing us to the subject and to the cyclic algorithm.


\begin{thebibliography}{99}

\bibitem{BDLLS} A. Barak, Z. Drezner, E. Levy, M. Lieber and A. Shiloh. \emph{Resilient gossip algorithms for collecting online management information in exascale clusters}. Concurrency and Computation: Practice and Experience, Vol. 27(17):4797-4818, Dec. 2015.

\bibitem{AB} A. Barak, private discussion.

\bibitem{C}  P. Berman, K. Diks and A. Pelc, \emph{Reliable broadcasting in logarithmic time with Byzantine link failures}, Journal of Algorithms 22 (1997), 199-211.

\bibitem{A} B.S. Chlebus, K. Diks and A. Pelc, \emph{Sparse networks supporting efficient reliable broadcasting}, Nordic Journal of Computing, Vol 1, No. 3 (1994). pp. 332-345.

\bibitem{G} K. Diks and A. Pelc, \emph{Efficient gossiping by packets in networks with random faults}, SIAM Journal on Discrete Mathematics 9 (1996), pp 8-17.

\bibitem{B}  K. Diks and A. Pelc, \emph{Almost safe gossiping in bounded degree networks}, SIAM Journal on Discrete Mathematics 5 (1992), pp. 338-344.

\bibitem{DP} K. Diks and A. Pelc, \emph{Optimal Adaptive Broadcasting with a Bounded Fraction of Faulty Nodes}, Algorithmica 28 (2000) Issue 1, pp 37–50

\bibitem{DDMM} B. Doerr, C. Doerr, S. Moran and S. Moran \emph{Simple and optimal randomized fault-tolerant rumor spreading}, Distributed Computing 29 (2016) Issue 2, pp 89–104.

\bibitem{DHL} B. Doerr, A. Huber and A. Levavi2\emph{Strong robustness of randomized rumor spreading protocols}, Discrete Applied Mathematics 161  (2013), Issue 6, pp 778-793

\bibitem{FG} A. M. Frieze and G. R. Grimmett, \emph{The shortest-path problem for graphs with random arc-lengths}, Discrete Applied Mathematics 10 (1985), 57-77.

\bibitem{GP} L. Gasieniec and A. Pelc, \emph{Adaptive broadcasting with faulty nodes}, Parallel Computing 22 (1996), 903-912

\bibitem{Z} T.H. Harris, \emph{The Theory of Branching processes}, Dover Publications, Inc., Mineola, NY, 2002.

\bibitem{F} E. Levy, A. Barak, A. Shiloh, M. Lieber, C. Weinhold and H. Hertig. \emph{Overhead of a decentralized gossip algorithm on the performance of HPC applications} Proc. Intr. Workshop on Runtime and Operating Systems for Supercomputers (ROSS), Munich, June 2014.

\bibitem{Y} R. Lyons, R. Pemantle and Y. Peres, \emph{Conceptual proofs of LlogL Criteria for mean behavior of branching processes},  Ann. Probab. 23 (1995), no. 3, 1125-1138.

\bibitem{D} P. Panaite and A. Pelc, \emph{Optimal broadcasting in faulty trees}, Journal of Parallel and Distributed Computing 60 (2000), 566-584.

\bibitem{Pelc1} A. Pelc, \emph{Fault-tolerant broadcasting and gossiping in communication networks}, Networks 28 (1996), 143-156.

\bibitem{Pittel} B. Pittel, \emph{On spreading a rumor}, SIAM Journal on Applied Mathematics 47 (1987), 213-223.

\end{thebibliography}
\end{document}